\documentclass[a4paper,twocolumn]{article}   
\usepackage{geometry}
\usepackage{amssymb,amsmath,amsthm}
\usepackage{authblk}
\usepackage{graphicx}
\usepackage{gunther2e}
\DeclareMathOperator*{\avg}{Avg}

\title{\textbf{Convergent Under-Approximations of Reachable Sets and Tubes for  Linear Uncertain Systems}} 
\author{Mohamed Serry
 \thanks{Mohamed Serry is with the Department of Mechanical and Mechatronics Engineering, University of Waterloo, Waterloo, Ontario, Canada.}}   
\date{}   
\begin{document}
\maketitle
\begin{abstract}              
 In this note, we propose a  method to under-approximate finite-time reachable sets and tubes for a class of continuous-time linear uncertain systems. The class under consideration is the  linear time-varying (LTV) class with integrable  time-varying system matrices and uncertain initial and input values belonging to known convex compact sets. The proposed method depends upon the iterative use of constant-input reachable sets which results in convergent under-approximations in the sense of Hausdorff distance. We illustrate our approach through two numerical examples.
\end{abstract}

\section{Introduction}
 In  recent years, the field of reachability analysis has received intensive research attention due to its significant role in  formal verification and control synthesis \cite{Althoff10,i14sym}.
 This resulted in a literature rich in various methods to approximate reachable sets and tubes for various classes of systems \cite{AsarinDangFrehseGirardLeGuernicMaler06}. 
 In the context of reachability analysis, approximation techniques can be classified into over-approximation approaches that envelope reachable sets and tubes (see, e.g.,  \cite{Althoff10,LeGuernicGirard09,ScottBarton12,SerryReissig18aC})
 and under-approximation approaches that inner bound reachable sets and tubes (see, e.g., \cite{GoubaultPutot17,KurzhanskiVaraiya00b,Varaiya98,XueSheEaswaran16}). In comparison to over-approximation methods in the literature, under-approximation approaches are relatively  underdeveloped despite their useful role in synthesis and verification (see, e.g., \cite{GoubaultPutot17,XueSheEaswaran16}). This motivates us in this work to develop an  under-approximation method for a class of linear uncertain systems.
 
Up to the knowledge of the author, there have been very few works on
 under-approximation methods for linear systems. For example, polytopic  and ellipsoidal approaches have been proposed  in \cite{Varaiya98}  and \cite{KurzhanskiVaraiya00b}, respectively, which rely on solving several initial value problems to obtain tight under-approximations that touch reachable sets at some boundary points. 
 Besides that, approaches designed for nonlinear systems can be used for linear systems, with the drawback of relatively high computational costs, where interval arithmetic \cite{GoubaultPutot17} and coverings of boundaries of reachable sets \cite{XueSheEaswaran16} are utilized.

In \cite{Veliov92}, second-order approximations  of reachable sets have been proposed, for linear uncertain systems, which rely on the use of local third-order approximations of constant-input reachable sets. In this work, we aim at utilizing constant-input reachable sets, due to their under-approximating nature, to propose a convergent under-approximation method for linear uncertain systems. Moreover, the proposed method is motivated  by the works in \cite{Althoff10,GirardLeGuernicMaler06,LeGuernicGirard09}, where efficient  over-approximations have been introduced for linear time-invariant (LTI) uncertain systems.

The contribution of the current work over the work in \cite{Veliov92} is as follows. First,  the proposed method in the current work is used to under-approximate both reachable sets and tubes rather than only approximate final-time reachable sets. Moreover,  the assumptions imposed on the linear systems under consideration are significantly weaker than  the ones in \cite{Veliov92}. 
The weakened assumptions are motivated by the need of convergent under-approximation methods that are applicable for a wide variety of linear systems ( e.g.,  switched linear systems \cite{Sun05}).

The organisation of this note is as follows. After the introduction, we present the necessary mathematical preliminaries. Then, we introduce the LTV system under consideration, the associated assumptions, and the problem statement. After that, we present the proposed method and prove its  convergence. Moreover, we discuss the implementability of our approach. Finally, the proposed method is demonstrated through two numerical examples.

\section{Preliminaries}
$\mathbb{R}$, $\mathbb{R}_+$, $\mathbb{Z}$ and $\mathbb{Z}_{+}$ denote
the sets of real numbers, non-negative real numbers, integers and
non-negative integers, respectively, and
$\mathbb{N} = \mathbb{Z}_{+} \setminus \{ 0 \}$. Given $a,b \in \mathbb{R}$, with $a\leq b$, $\intcc{a;b}$ denotes the discrete interval $\intcc{a,b}\cap \mathbb{Z}$. Given  $S\subseteq \mathbb{R}$, $\abs{S}$ denotes the Lebesgue measure of $S$ \cite{RoydenFitzpatrick88}; e.g., $\abs{\intcc{a,b}}=b-a$.
Arithmetic operations involving subsets of a linear space $X$ are
defined point-wise, e.g.
$\alpha M \defas \Set{ \alpha y }{ y \in M }$ and the Minkowski sum
$M + N \defas \Set{ y + z }{ y \in M, z \in N }$, if
$\alpha \in \mathbb{R}$ and $M, N \subseteq X$. We denote the identity map $X \to X \colon x \mapsto x$ by $\id$,
where the domain of definition $X$ will always be
clear form the context. By $\| \cdot \|$ we denote any norm on $X$, the norm of a non-empty
subset is defined by $\| M \| \defas \sup_{x \in M} \| x \|$. $d_{H}$ denotes Hausdorff distance (see \cite[Chapter~1,~Section~5,~p.~65]{AubinCellina84}). 
Given norms on $\mathbb{R}^n$ and $\mathbb{R}^m$, the linear space
$\mathbb{R}^{n \times m}$ of $n \times m$ matrices is endowed with the
induced matrix norm,
$\| A \| = \sup_{\| x \| \le 1} \| A x \|$ for
$A \in \mathbb{R}^{n \times m}$.
Given a non-empty set $X \subseteq \mathbb{R}^n$ and a compact
interval $\intcc{a,b} \subseteq \mathbb{R}$, $X^{\intcc{a,b}}$ denotes the set of Lebesgue measurable maps with domain  $\intcc{a,b}$ and  values in $X$.
Integration is
always understood in the sense of Bochner, an extension of Lebesgue integration \cite{HytonenVanNeervenVeraarWeis16}; e.g., a function $f\colon \intcc{a,b}\rightarrow \mathbb{R}^{n}$ is Bochner integrable, or simply integrable, if $f$ is measurable and $\norm{f(\cdot)}$ is Lebesgue integrable.  Almost every (where) is abbreviated as a.e \cite{RoydenFitzpatrick88}.
Given a non-empty $W\subseteq \mathbb{R}^{m}$, a compact interval $\intcc{a,b}$, and an integrable matrix-valued  function
$F\colon \intcc{a,b} \to \mathbb{R}^{n\times m}$, $\int_{a}^{b} F(t) W dt$ denotes the set-valued integral
$
\Set{ \int_{a}^{b} F(t) w(t) dt }{ w \in W^{\intcc{a,b}} }. 
$
\section{Problem Formulation}
In this section, we formulate and state the problem under consideration in this note.
\subsection{System Description}
Consider the LTV system
\begin{equation}\label{eq:LinearSystem}
\dot{x}=A(t)x+B(t)u(t),
\end{equation}
 over the time interval $\intcc{\underline{t},\overline{t}},~-\infty<\underline{t}<\overline{t}<\infty$, where $x(t)\in \mathbb{R}^{n}$ is the system state, $u(t)\in \mathbb{R}^{m}$ is the input, and $A\colon \intcc{\underline{t},\overline{t}}\rightarrow  \mathbb{R}^{n\times n}$ and $B\colon\intcc{\underline{t},\overline{t}}\rightarrow  \mathbb{R}^{n\times m}$  represent the time-varying system matrices. The initial value $x(\underline{t})$ , and the input $u(t)$ are subject to uncertainties. Assume $A(\cdot)$ is integrable. Given an initial value $x(\underline{t})=x_{0}$ and an input signal $u\colon \intcc{\underline{t},\overline{t}}\rightarrow \mathbb{R}^{m}$, such that $B(\cdot)u(\cdot)$ is integrable, the unique solution, $\varphi(\cdot,\underline{t},x_0,u)$, to system \ref{eq:LinearSystem}, generated by  $x_0$ and $u(\cdot)$, on $\intcc{\underline{t},\overline{t}}$ is given by \cite[Theorem~6.5.1,~p.~114]{Lukes82}
\begin{equation}\label{eq:solution}
\varphi(t,\underline{t},x_0,u)=\phi(t,\underline{t})x_{0}+\int_{\underline{t}}^{t}\phi(t,s)B(s)u(s)ds,
\end{equation}
$t\in \intcc{\underline{t},\overline{t}}$. Here,  $\phi\colon \intcc{\underline{t},\overline{t}}\times \intcc{\underline{t},\overline{t}}\rightarrow \mathbb{R}^{n\times n}$ is the state transition matrix-valued function, which is continuous on $\intcc{\underline{t},\overline{t}}\times \intcc{\underline{t},\overline{t}}$ and absolutely continuous in each variable, satisfying $\partial_{t}\phi(t,s)=A(t)\phi(t,s),~a.e.~ t\in \intcc{\underline{t},\overline{t}}$ and $\partial_{s}\phi(t,s)=-\phi(t,s)A(s),~a.e.~s\in \intcc{\underline{t},\overline{t}}$  \cite[Theorem~6.3.2,~p.~109]{Lukes82}. Moreover, $\phi(t,z)\phi(z,s)=\phi(t,s)$ and $\phi(s,s)=\id$ for all $t,s,z\in \intcc{\underline{t},\overline{t}}$.
Finally, $\phi$ satisfies the estimate 
\begin{equation}
\label{e:ExpEstimates}
\begin{split}
\| \phi(t,s) \| \le & \e^{ \int_{s}^{t} \norm{A(z)}dz }
\end{split}
\end{equation}
for all $\underline{t}\leq s\leq t\leq \overline{t}$ \cite[Lemma~6.3.1,~p.~108]{Lukes82}.

\subsection{Assumptions}\label{subseq:Assumptions}
 \begin{enumerate}[(i)]
   \setlength\itemsep{0in}
 \item The time interval $\intcc{\underline{t},\overline{t}}$ is compact and of non-zero length (i.e., $\abs{\intcc{\underline{t},\overline{t}}}\neq 0$).
 \item $A(\cdot)$ is integrable. 
 \item $B(\cdot)$ is measurable satisfying   $\int_{\underline{t}}^{\overline{t}}\norm{B(s)}^{p}ds<\infty$ for some $p\in (1,\infty]$. Note that under the imposed assumption, $B(\cdot)$ is integrable \cite[Corollary~3,~p.~142]{RoydenFitzpatrick88}. 
 \item The uncertain initial value $x(\underline{t})$ and input values $u(t),~t\in \intcc{\underline{t},\overline{t}}$ are assumed to belong to known sets $X_{0}$ and $U$, respectively,
i.e.,
$
x(\underline{t})\in X_{0},~u(t)\in U,~t\in \intcc{\underline{t},\overline{t}}
$.
\item $X_{0}$ and $U$ are non-empty, convex, and compact. 
 \end{enumerate}

\textbf{The problem data  $\underline{t}$ , $\overline{t}$, $A(\cdot)$, $B(\cdot)$, $X_0$ and $U$ are fixed  and the associated assumptions hold throughout this note.}

\subsection{Problem Statement}
Let $\mathcal{R}(t)$ denote \textbf{the reachable set} of system $\ref{eq:LinearSystem}$ at time $t\in \intcc{\underline{t},\overline{t}}$, with starting time $\underline{t}$, initial values in $X_{0}$, and input signals with values in $U$. In other words, 
$
\mathcal{R}(t)= \Set{\varphi(t,\underline{t},x_0,u)}{x_{0}\in X_0,~u(\cdot)\in U^{\intcc{\underline{t},\overline{t}}}},
$
$t\in \intcc{\underline{t},\overline{t}}$, where $\varphi$ is defined as in \ref{eq:solution}. Moreover, let $\mathcal{R}{(\intcc{\underline{t},\overline{t}})}$ denote \textbf{the reachable tube} of system \ref{eq:LinearSystem} over the time interval $\intcc{\underline{t},\overline{t}}$, with initial time $\underline{t}$, initial values in $X_{0}$ and input signals with values in $U$ (i.e., $\mathcal{R}(\intcc{\underline{t},\overline{t}})=\bigcup_{t\in \intcc{\underline{t},\overline{t}}}\mathcal{R}(t)$). Given the problem
data and a time discretization parameter $N$, design a method that  yields
 subsets $\mathcal{S}_{N}\subseteq \mathcal{R}(\overline{t})$ and $\mathcal{T}_{N}\subseteq \mathcal{R}(\intcc{\underline{t},\overline{t}})$  satisfying 
$ \mathcal{S}_{N} \to \mathcal{R}(\overline{t})$ and $ \mathcal{T}_{N} \to \mathcal{R}(\intcc{\underline{t},\overline{t}})$, in the sense of Hausdorff distance,
as $N \to \infty$.
\section{Under-Approximations}
In the section, we address the problem stated in the previous section by proposing a convergent under-approximation method.
\subsection{Proposed Method}
In this subsection, we introduce the proposed method and provide a thorough explanation of it.
Let $N\in \mathbb{N}$, and define the sets $\{\Lambda_{i}^{N}\}_{i=0}^{N}$ as follows:
\begin{equation}
\label{eq:ProposedMethod}
 \Lambda_{0}^{N}=X_{0},~\Lambda_{i}^{N}=\phi(t_{i},t_{i-1})\Lambda_{i-1}^{N}+ W_{i},~i\in \intcc{1;N},
\end{equation}
where
\begin{equation}\label{e:t&W}
\begin{split}
t_{i}&=\underline{t}+i{(\overline{t}-\underline{t})}/{N},~i\in \intcc{0;N},\\ W_{i}&=\left(\int_{t_{i-1}}^{t_{i}}\phi(t_{i},s)B(s)ds\right) U,~i\in \intcc{1;N}.
\end{split}
\end{equation}
Each  $\Lambda_{i}^{N}$, $i \in \intcc{1;N}$,  corresponds to the reachable set at time $t_{i}$, with initial time $t_{i-1}$, initial values in $\Lambda_{i-1}^{N}$,  and \underline{constant} input signals with values in $U$. Therefore, $\{\Lambda_{i}^{N}\}_{i=0}^{N}$  and  $\bigcup_{i=0}^{N}\Lambda_{i}^{N}$ under-approximate
$\{\mathcal{R}(t_{i})\}_{i=0}^{N}$ and
$\mathcal{R}(\intcc{\underline{t},\overline{t}})$, respectively. Convergence can be explained, \underline{informally}, as follows. The iterative use of constant inputs in the construction of  $\{\Lambda_{i}^{N}\}_{i=0}^{N}$ results in reachable sets under \underline{step} input signals. As the value of $N$ increases, the step input signals  approximate  measurable input signals more accurately 
and therefore, $\{\Lambda_{i}^{N}\}_{i=0}^{N}$ converge to the exact reachable sets $\{\mathcal{R}(t_{i})\}_{i=0}^{N}$. The set  $\bigcup_{i=0}^{N}\Lambda_{i}^{N}$ can be thought of as a set-valued step approximation of the set-valued function $\mathcal{R}(\cdot)$. Due to the continuity of  $\mathcal{R}(\cdot)$, in the sense of Hausdorff distance, the accuracy of the step approximation increases as $N$  increases which implies the convergence of $\bigcup_{i=0}^{N}\Lambda_{i}^{N}$ to $\mathcal{R}(\intcc{\underline{t},\overline{t}})$.
\subsection{Main Results}
Now, we introduce theorems \ref{th:UnderApproximationReachSet} and \ref{th:UnderApproximationReachTube}, which are the main results of this work,  to validate the under-approximations obtained by \ref{eq:ProposedMethod} and to illustrate their first order convergence.
\begin{theorem}[Convergent Under-Approximations of Reachable Sets]\label{th:UnderApproximationReachSet}
Let $N$ be a positive integer and $\{\Lambda_{i}^{N}\}_{i=0}^{N}$ and $\{t_{i}\}_{i=0}^{N}$ be defined as in \ref{eq:ProposedMethod}. Then, for all $i\in \intcc{0;N}$,
$
\Lambda_{i}^{N}\subseteq \mathcal{R}(t_{i})
$
and
$
\Lambda_{i}^{N}\rightarrow \mathcal{R}(t_{i})
$ as $N\rightarrow \infty$.
\end{theorem}

\begin{theorem}[Convergent Under-Approximations of Reachable Tubes]\label{th:UnderApproximationReachTube}
Let $N$ be a positive integer and the sets $\{\Lambda_{i}^{N}\}_{i=0}^{N}$ be defined as in \ref{eq:ProposedMethod}. Then,  $\bigcup_{i=0}^{N}\Lambda_{i}^{N}\subseteq \mathcal{R}(\intcc{\underline{t},\overline{t}})$ and $\bigcup_{i=0}^{N}\Lambda_{i}^{N}\rightarrow \mathcal{R}(\intcc{\underline{t},\overline{t}})$ as $N\rightarrow \infty$.
\end{theorem}
To prove theorems \ref{th:UnderApproximationReachSet} and \ref{th:UnderApproximationReachTube}, we resort to the following result which implies the semi-group property of reachable sets.
\begin{lemma}\label{lem:SemiGroup}
Given $a,b\in \intcc{\underline{t},\overline{t}},~ a\leq b$, 
$
\mathcal{R}(b)=\phi(b,a) \mathcal{R}(a)+\int_{a}^{b}\phi(b,s)B(s)Uds.
$
\end{lemma}

\begin{proof}[Proof of Theorem \ref{th:UnderApproximationReachSet}]
Let $\tau=(\overline{t}-\underline{t})/N$. Define the set-valued maps $I(\cdot,\cdot)$ and $J(\cdot,\cdot)$ as  $I(b,a)=\int_{a}^{b}\phi(b,s)B(s)dsU$ and $J(b,a)=\int_{a}^{b}\phi(b,s)B(s)Uds$, where $\underline{t}\leq a\leq b \leq \overline{t}$ (note that, by definition $I(b,a)\subseteq J(b,a)$). For convenience, let $I_{i}$ and $J_{i}$ denote the sets $I(t_{i},t_{i-1})$ and $J(t_{i},t_{i-1})$, respectively, $i\in \intcc{1;N}$. The first claim of theorem \ref{th:UnderApproximationReachSet} holds for $i=0$ as $\Lambda_{0}^{N}=\mathcal{R}(\underline{t})=X_{0}$. By induction, assume the first claim holds for $i\in \intcc{0;N-1}$, then it holds  for $i+1$ in place of $i$ as, using lemma \ref{lem:SemiGroup},
\begin{equation}\nonumber
\begin{split}
\Lambda_{i+1}^{N}&= \phi(t_{i+1},t_{i})\Lambda_{i}^{N}+\overbrace{W_{i+1}}^{=I_{i+1}}\\
&\subseteq \phi(t_{i+1},t_{i})\mathcal{R}(t_{i})+J_{i+1}=\mathcal{R}(t_{i+1})  
\end{split}
\end{equation}
which proves the first claim. Now, we prove the second claim. 
First, note that each set $I_{i},~i\in \intcc{1;N}$, can be written as a set-valued integral as follows: 
\begin{equation}\nonumber
I_{i}=\left(\frac{1}{\tau}\int_{t_{i-1}}^{t_{i}}\phi(t_{i},s)B(s)ds\right)\overbrace{(\tau U)}^{=\int_{t_{i-1}}^{t_{i}}Uds}=\int_{t_{i-1}}^{t_{i}}\mathcal{L}_{i}Uds,
\end{equation}
where $\mathcal{L}_{i}=(1/\tau)\int_{t_{i-1}}^{t_{i}}\phi(t_{i},s)B(s)ds$. Define $\gamma_{i}= d_{H}(\mathcal{R}(t_{i}),\Lambda_{i}^{N}),~i\in \intcc{0,N}$ and $\alpha_{i}=d_{H}(I_{i},J_{i}),~i\in \intcc{1;N}$. Then, using lemma \ref{lem:SemiGroup} and estimate \ref{e:ExpEstimates},
 we have
\begin{equation}\nonumber
\begin{split}
\gamma_{i}\leq& \norm{\phi(t_{i},t_{i-1})} \gamma_{i-1}+ d_{H}(I_{i},J_{i})\\
& \leq \e^{\int_{t_{i-1}}^{t_{i}}\norm{A(z)}dz}\gamma_{i-1}+\alpha_{i},~i\in \intcc{1;N},
\end{split}
\end{equation}
 with $\gamma_{0}=0$. It can be shown, using induction, that 
$
\gamma_{i}\leq \sum_{j=1}^{i}\e^{\int_{t_{j}}^{t_{i}}\norm{A(z)}dz}\alpha_{j},~i\in \intcc{0;N},
$
where $\sum_{j=1}^{0}(\cdot)=0$. Consequently, for all $~i\in \intcc{0;N}$, we have
\begin{equation}\label{eq:EstimateGamma}
\gamma_{i}\leq \e^{\int_{\underline{t}}^{\overline{t}}\norm{A(z)}dz}\sum_{j=1}^{N}\alpha_{j}.
\end{equation}
Let $C=\norm{U}\e^{\int_{\underline{t}}^{\overline{t}}\norm{A(z)}dz}$. By  utilizing estimate \ref{e:ExpEstimates} and properties of transition matrices, it can be shown that each $\alpha_{i},i\in \intcc{1;N}$,  satisfies 
\begin{equation}\nonumber
\alpha_{i}\leq C \int_{t_{i-1}}^{t_{i}}\left(\frac{1}{\tau}\int_{t_{i-1}}^{t_{i}}\norm{\psi(s)-\avg_{\intcc{t_{i-1},t_{i}}}(\psi)}ds\right) dz,
\end{equation}
where $\psi(\cdot)=\phi(\underline{t},\cdot)B(\cdot)$ and $\avg_{S}(\cdot)=(1/\abs{S})\int_{S}(\cdot)$ is the integral mean over a measurable subset $S$ (of finite non-zero measure). Let $\bar{\psi}\colon \mathbb{R}\rightarrow \mathbb{R}^{n\times m}$ be an extension of $\psi$ over the real line with zero values outside of $\intcc{\underline{t},\overline{t}}$, i.e.,
$$
\bar{\psi}(t)=\begin{cases}
\psi(t),&~t\in \intcc{\underline{t},\overline{t}},\\
0,& \text{otherwise}. 
\end{cases}
$$
Consequently, we have
\begin{equation}\label{eq:EstimateAlpha}
\alpha_{i}\leq  C \int_{t_{i-1}}^{t_{i}} M^{\#}_{\tau}[\bar{\psi}](s)ds,~i \in \intcc{1;N},
\end{equation}
where $M^{\#}_{\tau}[\bar{\psi}](\cdot)\colon \mathbb{R}\rightarrow \mathbb{R}_{+}$ is a truncated sharp maximal function (see, e.g.,  \cite{Grafakos09}) 
 which is defined as 
$$
M_{\tau}^{\#}[\bar{\psi}](x)=\sup_{S\ni x,~0<\abs{S}\leq \tau} \frac{1}{\abs{S}}\int_S \norm{\bar{\psi}(z)-\avg_{S}(\bar{\psi})}dz,
$$
where the supremum is taken over all \underline{open} intervals $S\subset \mathbb{R}$ that contain $x$, with lengths less than or  equal to $\tau$. Note that $M^{\#}_{\tau}[\bar{\psi}](\cdot)$ is lower semi-continuous and, consequently, measurable. Moreover, $M^{\#}_{\tau}[\bar{\psi}](\cdot)$ satisfies $\int_{\mathbb{R}}\abs{M^{\#}_{\tau}[\bar{\psi}](x)}^{p}dx<\infty$, where $p$ is the same as in subsection \ref{subseq:Assumptions}, assumption (iii). This claim is proved as  follows. Let $M[\bar{\psi}](\cdot)\colon \mathbb{R}\rightarrow \mathbb{R}_{+}$ be the Hardy-Littlewood maximal function  which is defined as
$$
M[\bar{\psi}](x)=\sup_{S\ni x,~0<\abs{S}<\infty} \frac{1}{\abs{S}}\int_S \norm{\bar{\psi}(z)}dz.
$$
Then, it can be easily verified that  $M^{\#}_{\tau}[\bar{\psi}](x)\leq 2 M[\bar{\psi}](x),~x\in \mathbb{R}$. Moreover, using Hardy-Littlewood maximal theorem \cite[Theorem~ 2.3.2,~p.~99]{HytonenVanNeervenVeraarWeis16}, we have $\int_{\mathbb{R}}\abs{M[\bar{\psi}](x)}^{p}dx<\infty$ as  $\int_{\mathbb{R}}\norm{\bar{\psi}(x)}^{p}dx<\infty$ which proves the claim. Consequently,  $M^{\#}_{\tau}[\bar{\psi}](\cdot)$ is integrable over $\intcc{\underline{t},\overline{t}}$ for all $0<\tau \leq \overline{t}-\underline{t}$. Using estimates \ref{eq:EstimateGamma} and \ref{eq:EstimateAlpha}, we have 
$$
\gamma _{i}\leq C\e^{\int_{\underline{t}}^{\overline{t}}\norm{A(z)}dz} \int_{\underline{t}}^{\overline{t}} M^{\#}_{\tau}[\bar{\psi}](s)ds,~i\in \intcc{0;N}.
$$
Note that 
$
M^{\#}_{\tau}[\bar{\psi}](x)\leq 2 T_{\tau}[\bar{\psi}](x),~x\in \mathbb{R}, 
$
where $T_{\tau}[\bar{\psi}](\cdot)\colon \mathbb{R}\rightarrow \mathbb{R}_{+}$ is defined as 
$$
T_{\tau}[\bar{\psi}](x)=
\sup_{S\ni x,~0<\abs{S}\leq \tau} \frac{1}{\abs{S}}\int_S \norm{\bar{\psi}(x)-\bar{\psi}(z)}dz.
$$
 Using Lebesgue differentiation theorem \cite[Theorem~2.3.4,~p.~101]{HytonenVanNeervenVeraarWeis16}, we have $
T_{\tau}[\bar{\psi}](s)\rightarrow 0,~a.e.~s\in \intcc{\underline{t},\overline{t}} 
$
as $\tau\rightarrow 0^{+}$ ($N\rightarrow \infty$). 
Therefore, using  squeeze theorem,
$
M^{\#}_{\tau}[\bar{\psi}](s)\rightarrow 0,~a.e.~s\in \intcc{\underline{t},\overline{t}} 
$
as $\tau\rightarrow 0^{+}$ ($N\rightarrow \infty$). 
Finally, using  dominated convergence theorem \cite[Proposition~1.2.5,~p.~16]{HytonenVanNeervenVeraarWeis16}, we have 
$
\int_{\underline{t}}^{\overline{t}} M^{\#}_{\tau}[\bar{\psi}](s)ds\rightarrow 0
$
as $N\rightarrow \infty$ which completes the proof.

\end{proof}
\begin{proof}[Proof of Theorem \ref{th:UnderApproximationReachTube}]
Let  $\{t_{i}\}_{i=0}^{N}$ be as in \ref{e:t&W},  $\beta(t,s)=\int_{s}^{t}\norm{B(s)}ds,~\underline{t}\leq s\leq t\leq \overline{t}$, and $M=\int_{\underline{t}}^{\overline{t}}\norm{A(z)}dz$. Using theorem \ref{th:UnderApproximationReachSet}, we have $\bigcup_{i=0}^{N}\Lambda_{i}^{N}$ $\subseteq$ $\bigcup_{i=0}^{N}\mathcal{R}(t_{i})$ $\subseteq$ $\mathcal{R}(\intcc{\underline{t},\overline{t}})$ which proves the first claim. Now, we prove the second claim. First we note, using the definition of reachable sets and \ref{e:ExpEstimates}, that
\begin{equation}\label{eq:Growthbound}
\norm{\mathcal{R}(t)}\leq K,~t\in \intcc{\underline{t},\overline{t}},
\end{equation}
where
$
K=\e^{M}\left(\norm{X_{0}}+\norm{U}\beta(\overline{t},\underline{t}) \right)
$.
The Hausdorff distance can be estimated as
\begin{equation}\label{eq:HausEsim}
d_{H}(\bigcup_{i=0}^{N} \Lambda_{i}^{N}, \mathcal{R}(\intcc{\underline{t},\overline{t}})) \leq\max_{i\in \intcc{0;N-1}} \sup_{t\in \intcc{t_{i},t_{i+1}}} d_{H}(\mathcal{R}(t),\Lambda_{i}^{N}).
\end{equation}
Let $\varepsilon>0$ be arbitrary. Choose $N\in \mathbb{N}$ sufficiently large such that:
\begin{flalign}
\label{eq:a1}
d_{H}(\Lambda_{i}^{N},\mathcal{R}(t_{i}))\leq \frac{\varepsilon}{2},~i\in \intcc{0;N},\\
\label{eq:a2}
\norm{\phi(t,t_{i})-\id}\leq \frac{\varepsilon}{4K},~t\in  [t_{i},t_{i+1}],~i\in \intcc{0;N-1},\\
\label{eq:a3}
\norm{U}\beta(t,t_{i}) \leq \frac{\varepsilon}{4  \e^{M}},~t\in [t_{i},t_{i+1}],~i\in \intcc{0;N-1}.
\end{flalign}
Such $N$ exists due to theorem \ref{th:UnderApproximationReachSet}, the \underline{uniform} continuity of $\phi$ on $\intcc{\underline{t},\overline{t}}\times\intcc{\underline{t},\overline{t}}$, and the fact that, using Holder's inequality \cite[Theorem~1,~p.~140]{RoydenFitzpatrick88} , $\beta(t,s)\leq (t-s)^{1/q}\left(\int_{\underline{t}}^{\overline{t}}\norm{B(s)}^{p}ds\right)^{1/p}$, where $q\in [1,\infty)$ satisfies $1/p+1/q=1$ and $p$ is the same as in subsection \ref{subseq:Assumptions}, assumption (iii). Let  $i\in \intcc{0;N-1}$ and 
$t\in \intcc{t_{i},t_{i+1}}$. Using the triangular inequality, we have  
\begin{equation}\label{eq:TriangularIneq}
d_{H}(\mathcal{R}(t),\Lambda_{i}^{N})
\leq 
d_{H}(\mathcal{R}(t),\mathcal{R}(t_{i}))
+
d_{H}(\mathcal{R}(t_{i}),\Lambda_{i}^{N})
\end{equation}
Using lemma \ref{lem:SemiGroup}, $\mathcal{R}(t)$ can be written as
$
\mathcal{R}(t)= \phi(t,t_{i})\mathcal{R}(t_{i})+ \int_{t_{i}}^{t}\phi(t,s)B(s)Uds.
$
Hence, using \ref{e:ExpEstimates}, \ref{eq:Growthbound}, \ref{eq:a2}, and \ref{eq:a3}, 
\begin{equation}\label{eq:d(t,ti)}
\begin{split}
d_{H}(R(t),\mathcal{R}(t_{i}))
\leq&\norm{\phi(t,t_{i})-\id}\norm{\mathcal{R}(t_{i})}\\
&+
\norm{U}\int_{t_{i}}^{t}\norm{\phi(t,s)}\norm{B(s)}ds\\
&\leq 
\norm{\phi(t,t_{i})-\id}K+\norm{U} \beta(t,t_{i})\e^{M}\\
& \leq \frac{\varepsilon}{4}+\frac{\varepsilon}{4}=\frac{\varepsilon}{2}.
\end{split}
\end{equation}
 Therefore, using \ref{eq:d(t,ti)}, \ref{eq:a1}, and \ref{eq:TriangularIneq}, 
$
d_{H}(\mathcal{R}(t),\Lambda_{i}^{N})\leq \frac{\varepsilon}{2}+\frac{\varepsilon}{2}= \varepsilon.
$
As $\varepsilon$ is arbitrary and in view of \ref{eq:HausEsim}, the proof is complete.
\end{proof}
\subsection{Implementation}
To enable implementing  the  proposed method, it is required that the initial and input sets, $X_{0}$ and $U$, belong to a class of sets that is closed under linear transformations and Minkowski sums. Such requirement is satisfied by the class of zonotopes (see, e.g., \cite[Section~2]{Althoff10}), and hence, the proposed method is implementable using zonotopes. Zonotopic approaches are well-established in reachability analysis due to the reasonable computational costs of the associated  Minkowski sums and linear transformations \cite{GirardLeGuernicMaler06}. For a zonotopic implementation of the proposed method, we refer the readers to, e.g., \cite[Section~3]{GirardLeGuernicMaler06}.
\section{Numerical Examples}
In this section, we demonstrate the proposed method via two numerical examples. 
For both examples, the closed-form formulas of the  transition matrices are known exactly. 
The proposed method is implemented using zonotopes with the help of CORA software \cite{Althoff15}.
\subsection{Academic Example}\label{subsec:Academic}
Consider the LTV system $\dot{x}=A(t)x+B(t)u(t)$
over the time interval $\intcc{\underline{t},\overline{t}}=\intcc{0,1}$, where $A(t)=\alpha(t)\id$,
$$
\alpha(t)=\begin{cases}0,&t=0,\\
\frac{1}{2\sqrt{t}},& \text{otherwise},
\end{cases}
$$
$$
B(t)=\e^{\sqrt{t}} \left(\begin{array}{cc}
    \cos(t) & -\sin(t) \\
     \sin(t)& \cos(t)
\end{array}\right),
$$
$x(\underline{t})\in X_{0}=\{(0,0)^{\intercal}\}$ and $u(t)\in U=\intcc{-1,1}\times \intcc{-1,1}
$. Note that $A(\cdot)$ is integrable (but not essentially bounded).
We  implement the proposed method, with different values of the time discretization parameter $N$, to obtain several under-approximations of $\mathcal{R}(\overline{t})$, $\Lambda_{N}^{N}$. \ref{fig:Academic} illustrates the obtained under-approximations which improve in accuracy as $N$ increases showing a converging behaviour as predicted by theorem \ref{th:UnderApproximationReachSet}.
\begin{figure}[htbp]
\centerline{
\includegraphics[height=2in
,keepaspectratio=true]{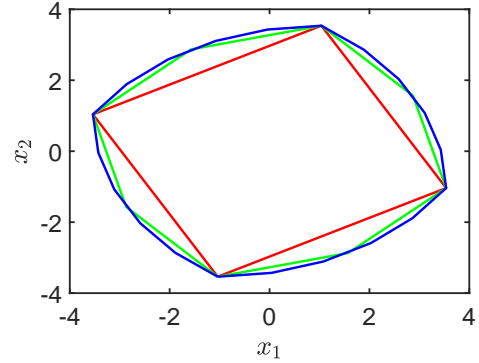}
} 
\centering
 \caption{Several under-approximations of $\mathcal{R}(\overline{t})$ (black), from example \ref{subsec:Academic}: $\Lambda^1_{1}$ (red), $\Lambda_{2}^{2}$ (green), and $\Lambda_{5}^{5}$ (blue).}
 \label{fig:Academic}
 \end{figure}

\subsection{Switched System}\label{subsec:DCDC}

In this example, we adopt a modified version of the DC-DC converter given \cite{GirardPolaTabuada10}.  Consider the LTV system 
$\dot{x}=A(t)x+u(t)$ on the time interval $\intcc{\underline{t},\overline{t}}=\intcc{0,5}$, where $x(\underline{t})\in X_{0}=\intcc{0.9,1.1}\times \intcc{4.9,5.1}$ and $u(t)\in U=\intcc{2/15
,8/15
}\times \{0\}$. $A(\cdot)$ is defined as

$$
A(t)=\begin{cases} A^{(1)},&~t\in \intcc{0,1}\cup\intcc{2,3},\\
A^{(2)},&~\text{otherwise},
\end{cases}
$$
where
$
A^{(1)}=\left(\begin{array}{cc} -\frac{1}{3} & 0\\ 0 & -\frac{1}{6} \end{array}\right)
$
and
$
A^{(2)}=    \left(\begin{array}{cc} -\frac{1}{2} & -\frac{1}{6}\\ \frac{1}{6} & -\frac{1}{6} \end{array}\right).
$
We use the proposed method, with several values of $N$, to obtain an under-approximations, $\bigcup_{i=0}^{N}\Lambda_{i}^{N}$, of $\mathcal{R}(\intcc{\underline{t},\overline{t}})$. \ref{fig:DCDC} illustrates several under-approximations of $\mathcal{R}(\intcc{\underline{t},\overline{t}})$. The mentioned figure displays the increased accuracy of the obtained under-approximations as $N$ increases which matches with the findings of theorem \ref{th:UnderApproximationReachTube}. 
\begin{figure}[htbp]
\centerline{
\includegraphics[height=1.3in
,keepaspectratio=true]{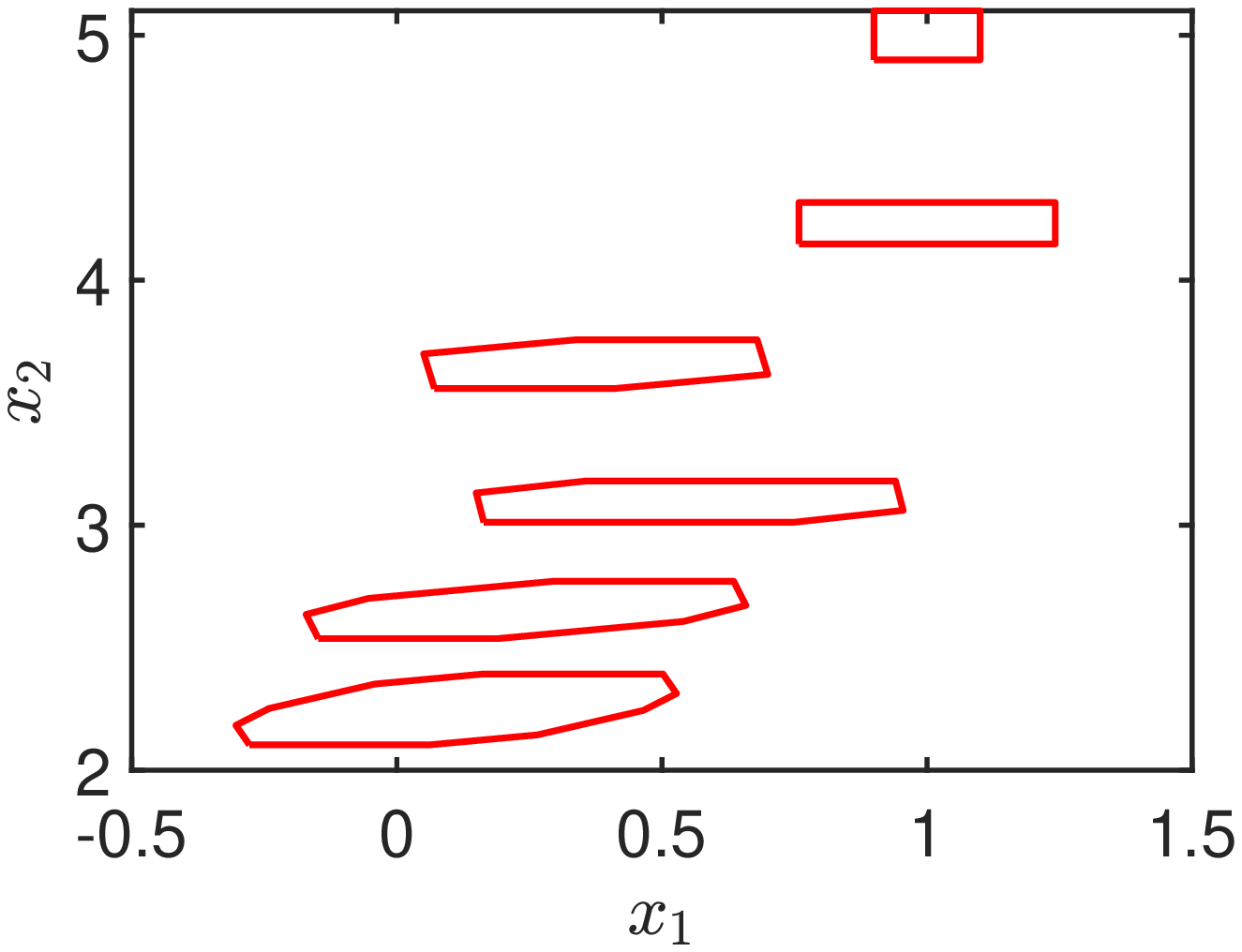}
\includegraphics[height=1.3in
,keepaspectratio=true]{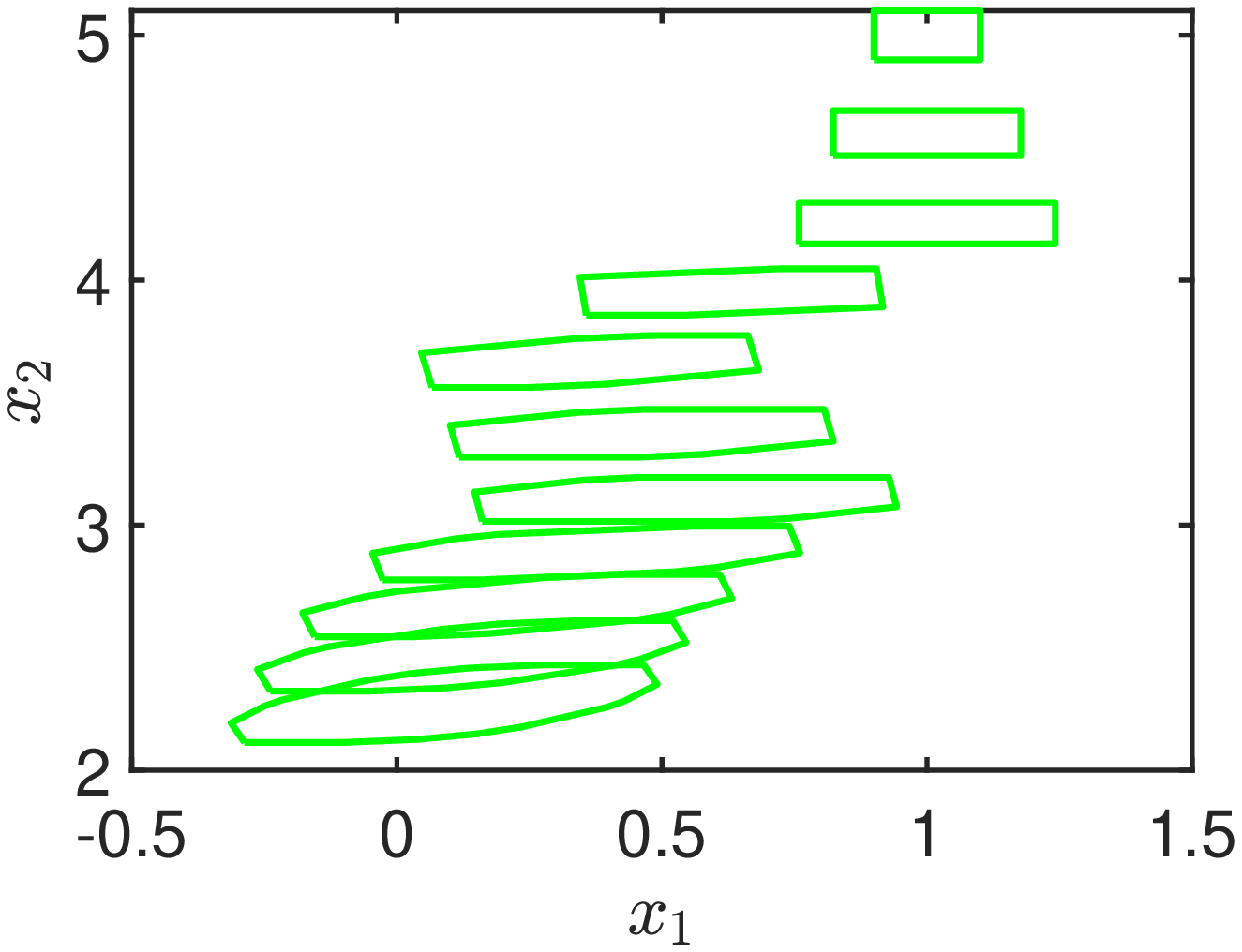}
}
\centerline{
\includegraphics[height=1.3in
,keepaspectratio=true]{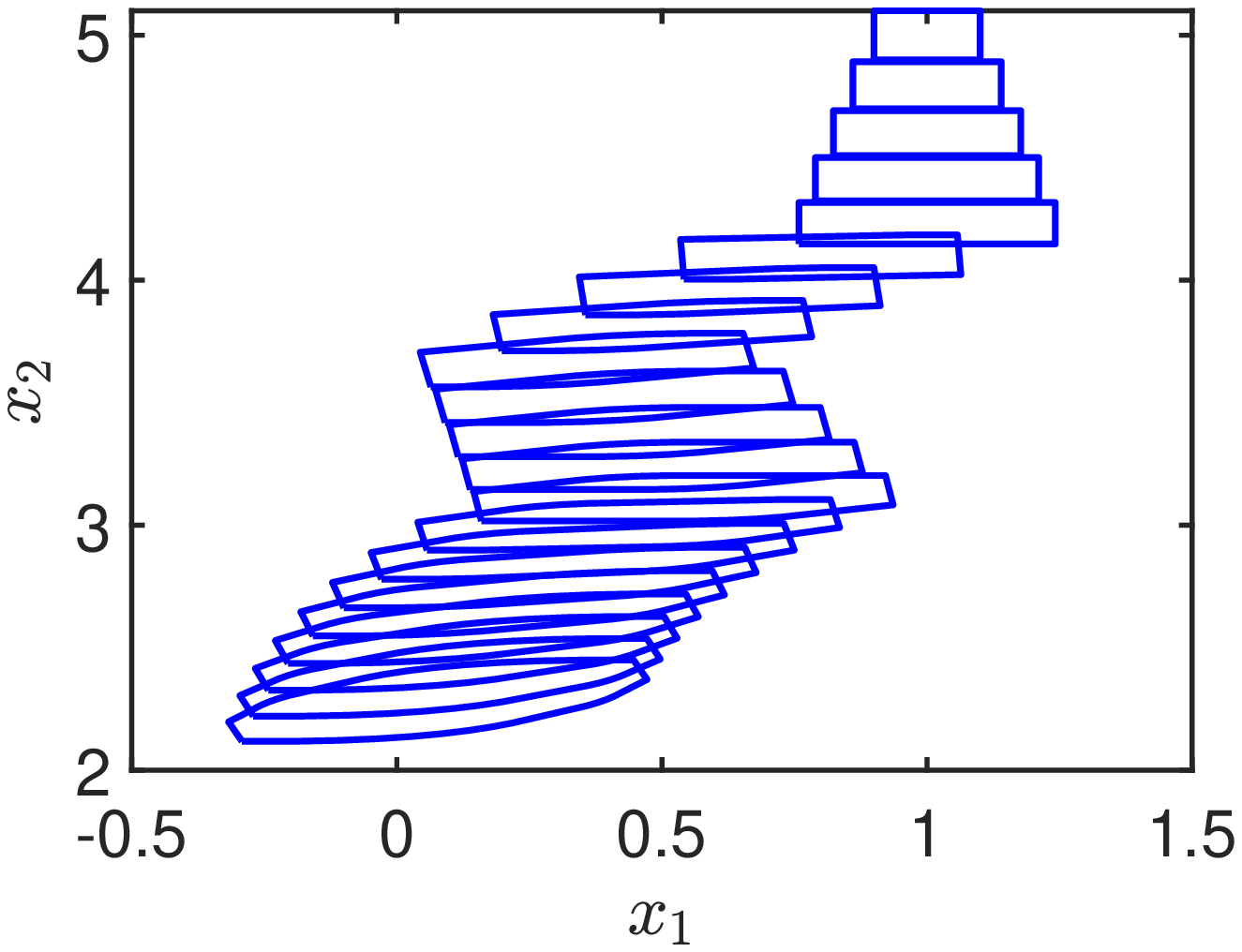}
\includegraphics[height=1.3in
,keepaspectratio=true]{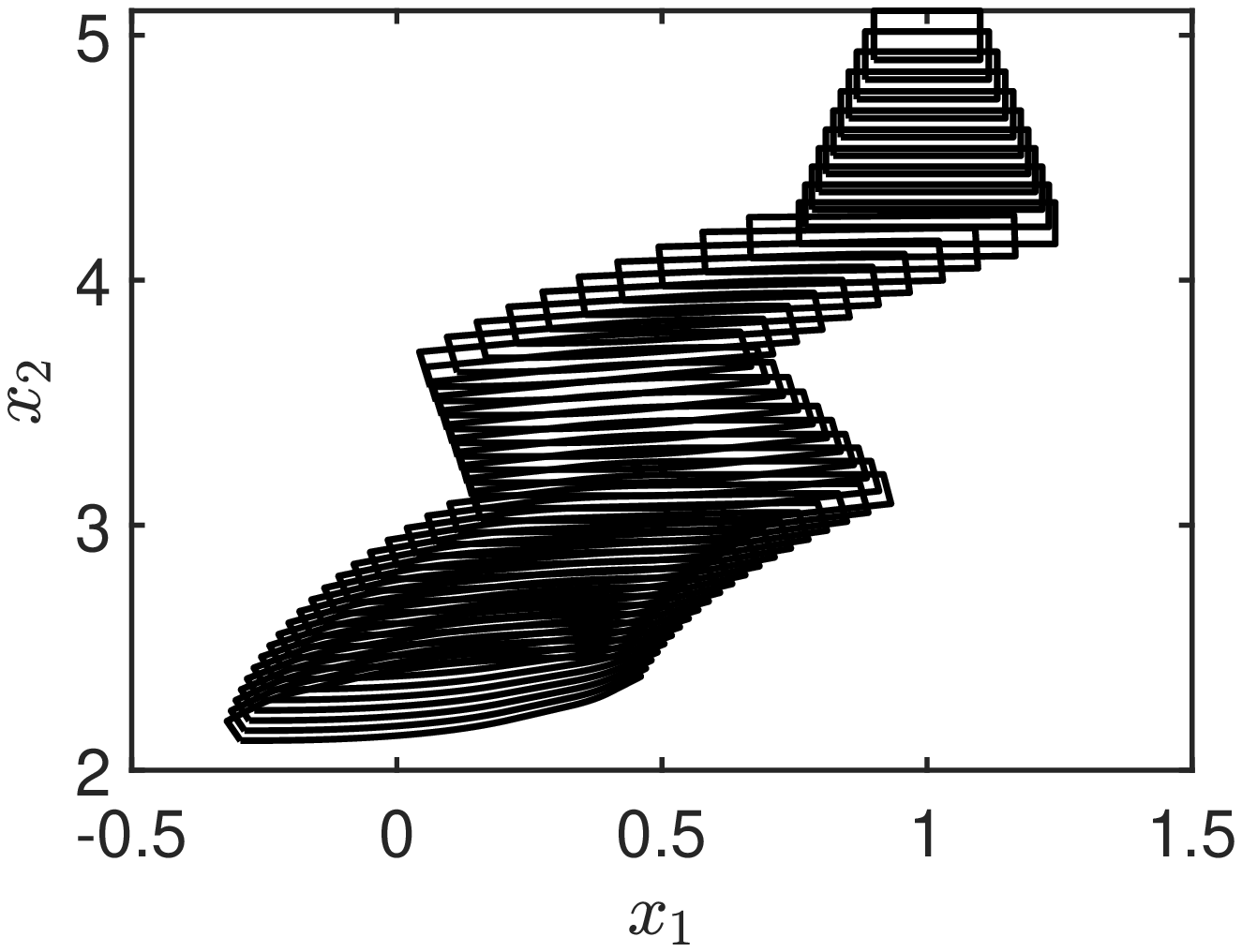}
}

\centering
 \caption{Under-approximations  of  $\mathcal{R}(\intcc{\underline{t},\overline{t}})$, from example \ref{subsec:DCDC}: $\bigcup_{i=0}^{5}\Lambda_{i}^{5}$ (red), $\bigcup_{i=0}^{10}\Lambda_{i}^{10}$ (green), $\bigcup_{i=0}^{20}\Lambda_{i}^{20}$ (blue), $\bigcup_{i=0}^{50}\Lambda_{i}^{50}$ (black). }
 \label{fig:DCDC}
 \end{figure}
\bibliographystyle{plain}        
\bibliography{IEEEtranBSTCTL,preambles,mrabbrev,strings,fremde,eigeneCONF,eigeneJOURNALS,eigenePATENT,eigeneREPORTS,eigeneTALKS,eigeneTHESES,fremdetmp,referencesstudents}

\end{document}